\documentclass[11pt,letterpaper,twoside,reqno]{amsart}

\usepackage{amsmath}
\usepackage{amssymb}
\usepackage{bm}
\usepackage{mathrsfs}
\usepackage[dvips]{graphicx}

\usepackage{color}

\usepackage{url}

\usepackage{supertabular}

\usepackage{alg}

\makeatletter
\def\algspacing{\alg@unmargin}
\makeatother

\newtheorem{thm}{Theorem}

\newtheorem{lemma}[thm]{Lemma}
\newtheorem{prop}[thm]{Proposition}

\newtheorem{defn}[thm]{Definition}

\newtheorem{fact}[thm]{Fact}

\theoremstyle{remark}
\newtheorem{rem}[thm]{Remark}

\DeclareMathAlphabet{\mathsfsl}{OT1}{cmss}{m}{sl}

\newcommand{\term}{\emph}

	{\makebox{\phantom{}} \\ %
		\textsc{Input:}\begin{itemize}}%
	{\end{itemize}}

	{\textsc{Output:}\begin{itemize}}%
	{\end{itemize}}

	{\textsc{Procedure:}\begin{enumerate}}%
	{\end{enumerate}}

\renewcommand{\phi}{\varphi}

\newcommand{\grad}{\nabla}

\newcommand{\mtx}[1]{\bm{#1}}

\newcommand{\trace}{\operatorname{tr}}

\newcommand{\psdgt}{\succ}

\newcommand{\ip}[2]{\left\langle {#1},\ {#2} \right\rangle}

\usepackage{subfigure}
\usepackage{wasysym}

\begin{document}
\title[Joint Convexity of Relative Entropy implies Lieb Concavity]
{From joint convexity of quantum relative entropy \\ 
to a concavity theorem of Lieb}

\author{Joel~A.~Tropp}




\date{27 December 2010. Revised 4 January 2011.}

\begin{abstract}
This note provides a succinct proof of a 1973 theorem of Lieb that establishes the concavity of a certain trace function.  The development relies on a deep result from quantum information theory, the joint convexity of quantum relative entropy, as well as a recent argument due to Carlen and Lieb.
\end{abstract}

\maketitle

\section{Introduction}
In his 1973 paper on trace functions, Lieb establishes an important concavity theorem~\cite[Thm.~6]{Lie73:Convex-Trace} concerning the trace exponential. 

\begin{thm}[Lieb] \label{thm:lieb}
Let $\mtx{H}$ be a fixed self-adjoint matrix.  The map
\begin{equation} \label{eqn:trace-fn}
\mtx{A} \longmapsto \trace \exp\left( \mtx{H} + \log \mtx{A} \right)
\end{equation}
is concave on the positive-definite cone.
\end{thm}

The most direct proof of Theorem~\ref{thm:lieb} is due to Epstein~\cite{Eps73:Remarks-Two}; see Ruskai's papers~\cite{Rus02:Inequalities-Quantum,Rus05:Erratum-Inequalities} for a condensed version of this argument.
Lieb's original proof 
develops the concavity of the function~\eqref{eqn:trace-fn} as a corollary of another deep concavity theorem~\cite[Thm.~1]{Lie73:Convex-Trace}.  In fact, many convexity and concavity theorems for trace functions are equivalent with each other, in the sense that the mutual implications follow from relatively easy arguments.  See~\cite[\S5]{Lie73:Convex-Trace} and \cite[\S5]{CL08:Minkowski-Type} for discussion of this point.

The goal of this note is to demonstrate that a modicum of convex analysis allows us to derive Theorem~\ref{thm:lieb} directly from another major theorem, the joint convexity of the quantum relative entropy.  The literature contains several elegant, conceptual proofs of the latter result; for example, see~\cite{Eff09:Matrix-Convexity}.  These arguments now deliver Lieb's theorem as an easy corollary.

The author's interest in Theorem~\ref{thm:lieb} stems from its striking applications in random matrix theory; refer to the paper~\cite{Tro10:User-Friendly-TR} for a detailed discussion.  Researchers concerned with these developments may find the current approach to Lieb's theorem more transparent than earlier treatments.

The main ideas in our presentation are drawn from the work of Carlen and Lieb~\cite{CL08:Minkowski-Type}, so this dispatch does not contain a truly novel technique.  Nevertheless, this note should be valuable because it provides a geometric intuition for Theorem~\ref{thm:lieb} and connects it to another major result.

\section{Background}

Our argument rests on the properties of a function, called the quantum relative entropy, which can be interpreted as a measure of dissimilarity between two positive-definite matrices.

\begin{defn} 
Let $\mtx{X}, \mtx{Y}$ be positive-definite matrices.  The \term{quantum relative entropy} of $\mtx{X}$ with respect to $\mtx{Y}$ is defined as
$$
{\rm D}( \mtx{X}; \mtx{Y} ) := 
	\trace( \mtx{X} \log \mtx{X} - \mtx{X} \log \mtx{Y} - (\mtx{X} - \mtx{Y})).
$$
Other appellations for this function include \term{quantum information divergence} and \term{von Neumann divergence}.
\end{defn}

The quantum relative entropy has a nice geometric interpretation~\cite[\S2.2 and \S2.6]{DT07:Matrix-Nearness}.  Define the quantum entropy function $\phi(\mtx{X}) := \trace(\mtx{X} \log \mtx{X})$ for a positive-definite argument.  The divergence  ${\rm D}(\mtx{X}; \mtx{Y})$ can be viewed as the difference between $\phi(\mtx{X})$ and the best affine approximation of the entropy $\phi$ at the matrix $\mtx{Y}$.  That is,
$$
{\rm D}( \mtx{X}; \mtx{Y}) = \phi(\mtx{X}) - [\phi(\mtx{Y}) + \ip{ \grad \phi(\mtx{Y}) }{ \mtx{X} - \mtx{Y} }].
$$
The entropy $\phi$ is a strictly convex function, which implies that the affine approximation strictly underestimates $\phi$.  This observation yields the following result.

\begin{fact} \label{fact:QID-nonneg} 
The quantum relative entropy is nonnegative:
$$
{\rm D}( \mtx{X}; \mtx{Y} ) \geq 0.
$$
Equality holds if and only if $\mtx{X} = \mtx{Y}$.
\end{fact}

\noindent
In quantum statistical mechanics, Fact~\ref{fact:QID-nonneg} is usually called \term{Klein's inequality}.  Another proof proceeds by showing that certain functional relations for scalars extend to matrix trace functions~\cite[\S2]{Pet94:Survey-Certain}.

The convexity properties of quantum relative entropy have paramount importance.  We require a major theorem, due to Lindblad~\cite[Lem.~2]{Lin74:Expectations-Entropy}, which encapsulates the difficulties of the proof.

\begin{fact}[Lindblad] \label{fact:QID-joint-convex} 
The quantum relative entropy is a jointly convex function.  That is,
$$
{\rm D}(t \mtx{X}_1 + (1-t) \mtx{X}_2; \, t\mtx{Y}_1 + (1-t) \mtx{Y}_2)
	\leq t \cdot {\rm D}(\mtx{X}_1; \mtx{Y}_1) + (1-t) \cdot {\rm D}(\mtx{X}_2; \mtx{Y}_2)
\quad\text{for $t \in [0,1]$},
$$
where $\mtx{X}_i$ and $\mtx{Y}_i$ are positive definite for $i = 1,2$.
\end{fact}

\noindent
Fact~\ref{fact:QID-joint-convex} follows easily from Lieb's main concavity theorem~\cite[Thm.~1]{Lie73:Convex-Trace}; Bhatia's book~\cite[\S IX.6 and Prob.~IX.8.17]{Bha97:Matrix-Analysis} offers a clear account of this approach.  The literature contains several other elegant proofs; 
see the papers~\cite{Uhl77:Relative-Entropy,PW78:Form-Convex,And79:Concavity-Certain,Han06:Extensions-Liebs}.  We single out Effros' work~\cite{Eff09:Matrix-Convexity} because it is accessible to researchers with experience in matrix theory and convex analysis.

Our final tool is a basic result from convex analysis which ensures that partial maximization of a concave function produces a concave function~\cite[Lem.~2.3]{CL08:Minkowski-Type}.  We include the simple proof.

\begin{prop} \label{prop:partial-max} 
Let $f( \, \cdot \,;  \,\cdot \, )$ be a jointly concave function.  Then the function $y \mapsto \max\nolimits_{x} f(x; y)$ obtained by partial maximization is concave, assuming the maximum is always attained.
\end{prop}

\begin{proof}
For each pair of points $y_1$ and $y_2$, there are points $x_1$ and $x_2$ that satisfy
$$
f(x_1; y_1) = \max\nolimits_x f(x; y_1)
\quad\text{and}\quad
f(x_2; y_2) = \max\nolimits_x f(x; y_2).
$$
For each $t \in [0,1]$, the joint concavity of $f$ implies that
\begin{align*}
\max\nolimits_x f(x; ty_1 + (1-t)y_2)
	&\geq f(t x_1 + (1-t)x_2; ty_1 + (1-t)y_2) \\
	&\geq t \cdot f(x_1; y_1) + (1-t) \cdot f(x_2; y_2) \\
	&= t \cdot \max\nolimits_x f(x; y_1) + (1-t) \cdot \max\nolimits_x f(x; y_2).
\end{align*}
In words, the partial maximum is a concave function.
\end{proof}

\section{Proof of Lieb's Theorem}

We begin with a variational representation of the trace, which is a restatement of the fact that quantum relative entropy is nonnegative.  The symbol $\psdgt$ denotes the positive-definite order.

\begin{lemma}[Variational Formula for Trace] \label{lem:variation}
Let $\mtx{Y}$ be a positive-definite matrix.  Then
$$
\trace \mtx{Y} = \max_{\mtx{X} \psdgt \mtx{0}} \ \trace( \mtx{X} \log \mtx{Y} - \mtx{X} \log \mtx{X} + \mtx{X}).
$$
\end{lemma}

\begin{proof}
Introduce the definition of the quantum relative entropy into Fact~\ref{fact:QID-nonneg} and rearrange to reach
$$
\trace \mtx{Y} \geq \trace( \mtx{X} \log \mtx{Y} - \mtx{X} \log \mtx{X} + \mtx{X}).
$$
When $\mtx{X} = \mtx{Y}$, both sides are equal, which yields the advertised result.
\end{proof}

The main result follows quickly using the variational formula and the other tools we have assembled.  As noted, the structure of this argument is parallel with the approach of Carlen--Lieb to another concavity theorem~\cite[Thm.~1.1]{CL08:Minkowski-Type}.

\begin{proof}[Proof of Theorem~\ref{thm:lieb}]
In the variational formula, Lemma~\ref{lem:variation}, select $\mtx{Y} = \exp(\mtx{H} + \log \mtx{A})$ to obtain
$$
\trace \exp( \mtx{H} + \log \mtx{A} )
	= \max_{\mtx{X} \psdgt \mtx{0}} \
	\trace(\mtx{X}(\mtx{H} + \log \mtx{A}) - \mtx{X} \log \mtx{X} + \mtx{X})
$$
The latter expression can be written compactly using the quantum relative entropy:
\begin{equation} \label{eqn:tr-exp-div}
\trace \exp( \mtx{H} + \log \mtx{A} )
	= \max_{\mtx{X} \psdgt \mtx{0}} \
	[ \trace(\mtx{XH}) - ({\rm D}(\mtx{X}; \mtx{A}) -  \trace \mtx{A})] 
\end{equation}
For each self-adjoint matrix $\mtx{H}$, Fact~\ref{fact:QID-joint-convex} implies that the bracket is a jointly concave function of the variables $\mtx{A}$ and $\mtx{X}$.  It follows from Proposition~\ref{prop:partial-max} that the right-hand side of~\eqref{eqn:tr-exp-div} defines a concave function of $\mtx{A}$.  This observation establishes the theorem.
\end{proof}

\begin{rem}
The expression~\eqref{eqn:tr-exp-div} states that the function $f : \mtx{H} \mapsto \trace \exp(\mtx{H} + \log \mtx{A})$ is the Fenchel conjugate of ${\rm D}( \,\cdot\,; \mtx{A} ) - \trace \mtx{A}$.  This observation implies that $f$ is convex.
\end{rem}

%
%
%

\section*{Acknowledgments}

The author thanks Eric Carlen for a very illuminating discussion of matrix convexity theorems, including the paper~\cite{CL08:Minkowski-Type}, as well as comments on an early draft of this paper.  Edward Effros contributed insights on quantum information theory, and Elliott Lieb emphasized the equivalences among concavity theorems.  This work has been supported in part by ONR awards N00014-08-1-0883 and N00014-11-1-0025, AFOSR award FA9550-09-1-0643, and a Sloan Fellowship.  The research was performed while the author attended the IPAM Fall 2010 program on optimization.

\bibliographystyle{alpha}
\bibliography{lieb}

\begin{thebibliography}{And79}

\bibitem[And79]{And79:Concavity-Certain}
T.~Ando.
\newblock Concavity of certain maps on positive definite matrices and
  applications to {H}adamard products.
\newblock {\em Linear Algebra Appl.}, 26:203--241, 1979.

\bibitem[Bha97]{Bha97:Matrix-Analysis}
R.~Bhatia.
\newblock {\em Matrix Analysis}.
\newblock Number 169 in Graduate Texts in Mathematics. Springer, Berlin, 1997.

\bibitem[CL08]{CL08:Minkowski-Type}
E.~A. Carlen and E.~H. Lieb.
\newblock A {M}inkowski-type trace inequality and strong subadditivity of
  quantum entropy {II}: {C}onvexity and concavity.
\newblock {\em Lett. Math. Phys.}, 83:107--126, 2008.

\bibitem[DT07]{DT07:Matrix-Nearness}
I.~S. Dhillon and J.~A. Tropp.
\newblock Matrix nearness problems with {B}regman divergences.
\newblock {\em SIAM J. Matrix Anal. Appl.}, 29(4):1120--1146, 2007.

\bibitem[Eff09]{Eff09:Matrix-Convexity}
E.~G. Effros.
\newblock A matrix convexity approach to some celebrated quantum inequalities.
\newblock {\em Proc. Natl. Acad. Sci. USA}, 106(4):1006--1008, Jan. 2009.

\bibitem[Eps73]{Eps73:Remarks-Two}
H.~Epstein.
\newblock Remarks on two theorems of {E.} {L}ieb.
\newblock {\em Comm. Math. Phys.}, 31:317--325, 1973.

\bibitem[Han06]{Han06:Extensions-Liebs}
F.~Hansen.
\newblock Extensions of {L}ieb's concavity theorem.
\newblock {\em J. Statist. Phys.}, 124(1):87--101, July 2006.

\bibitem[Lie73]{Lie73:Convex-Trace}
E.~H. Lieb.
\newblock Convex trace functions and the {W}igner--{Y}anase--{D}yson
  conjecture.
\newblock {\em Adv. Math.}, 11:267--288, 1973.

\bibitem[Lin74]{Lin74:Expectations-Entropy}
G.~Lindblad.
\newblock Expectations and entropy inequalities for finite quantum systems.
\newblock {\em Comm. Math. Phys.}, 39:111--119, 1974.

\bibitem[Pet94]{Pet94:Survey-Certain}
D.~Petz.
\newblock A survey of certain trace inequalities.
\newblock In {\em Functional analysis and operator theory}, volume~30 of {\em
  Banach Center Publications}, pages 287--298, Warsaw, 1994. Polish Acad. Sci.

\bibitem[PW78]{PW78:Form-Convex}
W.~Pusz and S.~L. Woronowicz.
\newblock Form convex functions and the {WYDL} and other inequalities.
\newblock {\em Lett. Math. Phys.}, 2:505--512, 1978.

\bibitem[Rus02]{Rus02:Inequalities-Quantum}
M.~B. Ruskai.
\newblock Inequalities for quantum entropy: {A} review with conditions for
  equality.
\newblock {\em J. Math. Phys.}, 43(9):4358--4375, Sep. 2002.

\bibitem[Rus05]{Rus05:Erratum-Inequalities}
M.~B. Ruskai.
\newblock Erratum: {I}nequalities for quantum entropy: {A} review with
  conditions for equality [\emph{{J}. {M}ath. {P}hys.} 43, 4358 (2002)].
\newblock {\em J. Math. Phys.}, 46(1):0199101, 2005.

\bibitem[Tro10]{Tro10:User-Friendly-TR}
J.~A. Tropp.
\newblock User-friendly tail bounds for sums of random matrices.
\newblock ACM Report 2010-01, California Inst. Tech., Pasadena, CA, Apr. 2010.

\bibitem[Uhl77]{Uhl77:Relative-Entropy}
A.~Uhlmann.
\newblock Relative entropy and the {W}igner--{Y}anase--{D}yson--{L}ieb
  concavity in an interpolation theory.
\newblock {\em Comm. Math. Phys.}, 54:21--32, 1977.

\end{thebibliography}

\end{document}